%% file: main.tex
\documentclass[letterpaper,twocolumn,10pt]{article}
\usepackage{usenix-2020-09}

\usepackage{hyperref}
\usepackage{cleveref}
\usepackage[square,numbers,sort&compress]{natbib}
\usepackage{color,balance}
\usepackage{xcolor,colortbl}
\usepackage{tikz}
\usepackage{xspace}
\usepackage{listings}
\usepackage{ctable}
\usepackage{subfig}
\usepackage{comment}
\usepackage{natbib}
\usepackage{amsthm}

\usepackage{url}

\newtheorem{theorem}{Theorem}[section]
\newtheorem{lemma}[theorem]{Lemma}

\definecolor{LightCyan}{gray}{0.9}
\definecolor{mGreen}{rgb}{0,0.6,0}
\definecolor{mRed}{rgb}{1.,0.0,0.0}
\definecolor{mPurple}{rgb}{0.58,0,0.82}
\definecolor{mBlue}{rgb}{0,0.0,1}
\definecolor{backgroundColour}{rgb}{0.95,0.95,0.92}
\lstdefinestyle{CStyle}{
	commentstyle=\color{mGreen},
	keywordstyle=\color{magenta},
	stringstyle=\color{mPurple},
	basicstyle=\ttfamily\small,	
	frame=tlbr,
	framesep=0.2cm, 
	framerule=0pt,
	breakatwhitespace=false,         
	breaklines=true,                 
	captionpos=b,                    
	keepspaces=true,                 
	showspaces=false,                
	showstringspaces=false,
	showtabs=false,                  
	tabsize=2,
	aboveskip=0pt,
    belowskip=0pt,
	language=C,
    literate={0}{{\textcolor{mPurple}{0}}}{1}%
         {1}{{\textcolor{mPurple}{1}}}{1}%
         {2}{{\textcolor{mPurple}{2}}}{1}%
         {3}{{\textcolor{mPurple}{3}}}{1}%
         {4}{{\textcolor{mPurple}{4}}}{1}%
         {5}{{\textcolor{mPurple}{5}}}{1}%
         {6}{{\textcolor{mPurple}{6}}}{1}%
         {7}{{\textcolor{mPurple}{7}}}{1}%
         {8}{{\textcolor{mPurple}{8}}}{1}%
         {9}{{\textcolor{mPurple}{9}}}{1}%
         {.0}{{\textcolor{mPurple}{.0}}}{2}%
         {.1}{{\textcolor{mPurple}{.1}}}{2}%
         {.2}{{\textcolor{mPurple}{.2}}}{2}%
         {.3}{{\textcolor{mPurple}{.3}}}{2}%
         {.4}{{\textcolor{mPurple}{.4}}}{2}%
         {.5}{{\textcolor{mPurple}{.5}}}{2}%
         {.6}{{\textcolor{mPurple}{.6}}}{2}%
         {.7}{{\textcolor{mPurple}{.7}}}{2}%
         {.8}{{\textcolor{mPurple}{.8}}}{2}%
         {.9}{{\textcolor{mPurple}{.9}}}{2}%
}

\lstdefinestyle{TeeSpecStyle}{
	commentstyle=\color{mGreen},
	keywordstyle=\color{magenta},
	stringstyle=\color{mPurple},
	basicstyle=\ttfamily\small,	
	frame=tlbr,
	framesep=0.2cm, 
	framerule=0pt,
	breakatwhitespace=false,         
	breaklines=true,                 
	captionpos=b,                    
	keepspaces=true,                 
	showspaces=false,                
	showstringspaces=false,
	showtabs=false,                  
	tabsize=2,	
	language=Python,
	aboveskip=0pt,
    belowskip=0pt,
	morekeywords={Map, relation, returns, action, fuzz, atomic, await, char, uint64\_t, size\_t, off\_t, string, safety, requires, extern, call},
    literate={0}{{\textcolor{mPurple}{0}}}{1}%
             {1}{{\textcolor{mPurple}{1}}}{1}%
             {2}{{\textcolor{mPurple}{2}}}{1}%
             {3}{{\textcolor{mPurple}{3}}}{1}%
             {4}{{\textcolor{mPurple}{4}}}{1}%
             {5}{{\textcolor{mPurple}{5}}}{1}%
             {6}{{\textcolor{mPurple}{6}}}{1}%
             {7}{{\textcolor{mPurple}{7}}}{1}%
             {8}{{\textcolor{mPurple}{8}}}{1}%
             {9}{{\textcolor{mPurple}{9}}}{1}%
             {.0}{{\textcolor{mPurple}{.0}}}{2}%
             {.1}{{\textcolor{mPurple}{.1}}}{2}%
             {.2}{{\textcolor{mPurple}{.2}}}{2}%
             {.3}{{\textcolor{mPurple}{.3}}}{2}%
             {.4}{{\textcolor{mPurple}{.4}}}{2}%
             {.5}{{\textcolor{mPurple}{.5}}}{2}%
             {.6}{{\textcolor{mPurple}{.6}}}{2}%
             {.7}{{\textcolor{mPurple}{.7}}}{2}%
             {.8}{{\textcolor{mPurple}{.8}}}{2}%
             {.9}{{\textcolor{mPurple}{.9}}}{2}%
}

\lstdefinestyle{grammar}{
	basicstyle=\ttfamily,	
	framesep=0.2cm,
	framerule=0pt,
	breakatwhitespace=false,         
	breaklines=true,                 
	captionpos=b,                    
	keepspaces=true,                 	
	showspaces=false,                
	showstringspaces=false,
	showtabs=false,    
	tabsize=2	
}

\newif\ifdisablenotes \disablenotestrue
\newcommand{\note}[2]{\textcolor{#2}{\textit{[#1]}}}
\ifdisablenotes
 \renewcommand{\note}[2]{}
\fi

\renewcommand{\paragraph}[1]{\vspace{0.02in}\noindent{\bf #1.}}

\lstdefinestyle{inline}{basicstyle={\sffamily\small\spaceskip=0.3em},columns=fullflexible}
\newcommand{\id}[1]{\lstinline[style=inline]+#1+}

\setupctable{doinside=\footnotesize}

\newcommand{\TeeSpec}{GKSpec\xspace}
\newcommand{\projname}{GateKeeper\xspace}
\newcommand{\BackendFuzz}{GKVulnChk\xspace} 
\newcommand{\BackendMock}{GKMock\xspace} 
\newcommand{\BackendVerify}{GKValidator\xspace} 

\newcommand{\GrapheneUs}{GK-Graphene\xspace} 
\newcommand{\GrapheneVanilla}{Vanilla-Graphene\xspace} 

\newcommand{\specification}{model\xspace}
\newcommand{\specifications}{models\xspace}
\newcommand{\Specification}{Model\xspace}

\begin{document}
\date{}

\title{Securing Access to Untrusted Services From TEEs with GateKeeper}

\author{
{\rm Meni Orenbach}\\
Technion
\and
{\rm Bar Raveh}\\
Technion
\and
{\rm Alon Berkenstadt}\\
Technion 
\and
{\rm Yan Michalevsky}\\
Anjuna Security
\and
{\rm Shachar Itzhaky}\\
Technion
\and
{\rm Mark Silberstein}\\
Technion
} 

\maketitle
\sloppy

\input{./abstract}
\input{./introduction} 
\input{./background} 
\input{./motivation} 
\input{./overview} 
\input{./design} 
\input{./evaluation} 
\input{./related} 
\input{./conclusion} 

\bibliographystyle{plainnat}
\bibliography{refs}

\end{document}

%% file: abstract.tex
\begin{abstract}
Applications running in Trusted Execution Environments (TEEs) commonly use untrusted external services such as host File System.
Adversaries 
may maliciously alter the normal service behavior to trigger subtle application bugs that would have never occurred under correct service operation, 
causing data leaks and integrity violations.
Unfortunately, existing manual protections are incomplete and ad-hoc, whereas
formally-verified ones require special expertise. 

We introduce \textbf{\projname}, a  framework to develop mitigations and vulnerability checkers for such attacks
by leveraging lightweight formal models of untrusted services.
With the attack seen as a violation of a services' functional correctness, 
\projname takes a novel approach to develop a comprehensive model of a service
without requiring formal methods expertise. We harness 
available testing suites routinely used in service development
to \emph{tighten} the model to known correct service implementation.
\projname uses the 
resulting model to automatically \emph{generate} (1) a correct-by-construction runtime service validator in C that is linked with a trusted application and guards each service invocation to conform to the model; 
and (2) a targeted model-driven vulnerability checker for analyzing black-box applications.

We evaluate \projname on  Intel SGX enclaves. 
We develop comprehensive \specifications of a POSIX file system and OS synchronization primitives while using thousands of existing test suites to tighten their models to the actual Linux implementations. 
We generate the validator and integrate it with Graphene-SGX, and successfully protect unmodified Memcached and SQLite with negligible overheads. 
The generated vulnerability checker detects novel vulnerabilities in the Graphene-SGX protection layer and production applications. 
\end{abstract}

%% file: introduction.tex
\section{Introduction}
\label{sec:introduction}
Trusted Execution Environments (TEE) are available in
several widely used CPUs from Intel, AMD, RISCV, and ARM~\cite{sgx:hasp13:mckeen, sgx:hasp13:anati,keystone:eurosys20,  amd_sev_snp:white_paper,keystone:eurosys20,komodo:sosp17,cca:white_paper} and are supported by public 
clouds ~\cite{markruss:azureconfidentialcomputing:2017,ibm:sgx:blog,alibabacloud:ebm:sgx}.
They protect the confidentiality and integrity of hosted code
and data, treating the rest of the system including privileged software as untrusted. 

Often, TEE applications use external untrusted services. For example, it is common to run unmodified applications in Intel SGX enclaves with the help of library OSs (libOSs), such as Haven, SCONE, and Graphene-SGX~\cite{haven:osdi14,scone:osdi16,sgxlkl,graphenesgx:usenix17}, which enable seamless invocation of untrusted OS services. Similarly, trusted virtual machines in AMD SEV and Intel TDX use services provided by an untrusted hypervisor~\cite{amd_sev_snp:white_paper,tdx:white_paper}.

The use of untrusted services inevitably creates security risks for TEE applications.
While data confidentiality and integrity protection are broadly deployed, a less obvious vector of attacks is via an untrusted interface: an attacker may deliberately manipulate arguments, return values, and semantics of the service to trigger application-related bugs that inadvertently cause data leakage or control-flow violations~\cite{iago:asplos13, besfs:usenixsec20, emilia:ndss21, tale_of_two_worlds:ccs19, via:arxiv}. For example, a futex may violate the mutual exclusion property~\cite{asyncshock:esorics16}, a file system may mix up file descriptors (\S\ref{sec:motivation:attacks}), or a virtual device may return an unexpected error code~\cite{via:arxiv}. We adopt the term \emph{Iago attacks}~\cite{iago:asplos13,tale_of_two_worlds:ccs19} to refer to all such interface attacks.

Defending against Iago attacks is an open challenge. In general, the mitigation entails validating the \emph{correctness} of all possible outcomes of each service invocation -- an increasingly difficult and error-prone task for complex stateful interfaces such as the POSIX file system (FS) API. Unfortunately, many production applications turned out to be vulnerable~\cite{emilia:ndss21}.  

To narrow the interface to an external service 
and make it easier to secure, one common solution is to reimplement part of the service inside the TEE. For example, SGX-LKL includes an in-enclave FS implementation that calls into the untrusted OS only to store and retrieve raw data blocks~\cite{sgxlkl}. 
Unfortunately, even a narrower block-based interface was found to be vulnerable~\cite{via:arxiv,tale_of_two_worlds:ccs19}.
Closest to our approach, the BesFS project develops an FS model and formally proves a set of correctness invariants on it. The model is used to validate that the underlying untrusted FS does not mount Iago attacks~\cite{besfs:usenixsec20}. However, its development requires extensive expertise in formal methods and  manual translation of the model into C to make it usable, which is harder to scale. 

{\bf \projname} provides a general and comprehensive solution to Iago-attack mitigation and vulnerability testing that is easy to use and deploy for a variety of services, and without assuming formal verification skills. Our target users are (1) service providers seeking to enable safe and secure access to their services from TEE programs; (2) TEE developers seeking to secure access to existing untrusted services; (3) TEE security analysts testing existing applications resilience to attacks.  

Inspired by the recent success of lightweight formal methods~\cite{lightweight_formal:sosp21}, 
we set to build our solution around an abstract service \emph{reference model}, 
which omits implementation details and represents only the visible functional behavior.

Given a model, \projname  \emph{automatically} generates a {trusted validator layer (\BackendVerify)} in C which can be readily integrated into an application or TEE runtime to securely invoke the service.  \BackendVerify embeds an executable model and performs runtime checks to ensure conformance to the expected service behavior expressed by the model. While the executable model is similar to a reference implementation used for property checking~\cite{lightweight_formal:sosp21}, it differs in that it invokes an untrusted service and validates its conformance to the model, \emph{without} implementing the service itself, which is particularly important when validating service \emph{behavior} (as we show in the case of mutex in~\S\ref{sec:motivation:model_protection}). 
Moreover, checking that the service output is the same as that of the model's is not enough, because multiple correct outputs could be allowed, as in the case of different POSIX-compliant FS implementations. \BackendVerify correctly handles these cases.

The same model is used to generate a targeted vulnerability checker (\BackendFuzz), which can serve to drive effective fuzzing sessions or be integrated into existing Iago vulnerability fuzzers~\cite{emilia:ndss21}. \BackendFuzz leverages the service model to check for deep vulnerabilities by automatically generating malicious values violating the model constraints tailored to the current state of the model (\S\ref{sec:overview:fuzz}).

The fundamental question remains: how to build a complete and correct functional model of a service? 
We address this issue by introducing a novel approach that specifically targets the prevention of Iago attacks for \emph{existing} services. 
We observe that to be effective against Iago attacks, the model must be \emph{tight} to a service high-level specification (i.e., POSIX for an FS) or an existing uncompromised widely-used implementation, e.g., ext4. We refer to such implementations as \emph{etalon implementations}.
Thus, we enable model development by using both the etalon implementation and existing regression and conformance testing  originally used for service development. This is particularly applicable to OS and hypervisor services which often provide extensive high-coverage testing suites. 

Our method works as follows. Given an initial version of a model, we test the generated \BackendVerify layer on top of the etalon implementation with the testing suites (this is done outside the TEE). 
Any failure implies that the model is over-constrained. On the other hand, we introduce a tool to \emph{automatically generate \BackendMock}, an executable mock of the service. 
\BackendMock is checked for correctness using the same testing suits, spotting additional bugs. Both \BackendMock and \BackendVerify can be regenerated effortlessly and tested again, until no new bugs are discovered. 
This method also makes it easy to adjust the model to the changes in the service API.  

Our approach does not aim to replace formal verification of a model, rather it is complementary. In fact, we made initial steps toward model verification by proving  safety of the mutex model by manually translating it into mypyvy~\cite{mypyvy:cav19}. However, even models with proven key safety properties do not guarantee complete model correctness~\cite{spec_guard:oospla17}. Thus, we believe that our approach can be useful to raise the confidence in the model correctness with modest development efforts, under a reasonable assumption of the availability of high-coverage service tests.

In summary, our contributions are: 

\noindent \textbf{Model and compiler(\S\ref{sec:design:lang})}. We introduce a simple C-like Domain Specific Language (DSL) for model development, which offers the primitives to specify the model, untrusted calls, and their conformance checks, including details needed for generating \BackendVerify, \BackendMock and \BackendFuzz by a compiler we develop (2,770 LOC).

\noindent \textbf{Models for an FS
and synchronization APIs  (\S\ref{sec:models})}. 
We develop the models of two complex OS interfaces while covering most of their APIs sufficient to run commodity applications. The models are much smaller than the full interface implementations: only 1,226 lines for the FS model and 300 lines for the futex and mutex models. Notably, writing the model is  intuitive: the mutex and futex models took roughly a month for a single undergraduate student.  

\noindent \textbf{Streamlined model development  (\S\ref{sec:overview})}. 
We validate that the models are tight to the POSIX specification and the respective Linux implementations, by generating  their executable mocks~(\S\ref{sec:design:spec_validation}). In particular, the generated FS mock is equivalent to a {\tt tmpfs} in-memory file system when mounted via the FUSE library. We test the FS model by running a high-coverage SibylFS POSIX conformance suite which includes over twenty thousand tests~\cite{sibylfs:sosp15}, and the synchronization model using an LTP~\cite{ltp} Linux stress-testing suites, 
obtaining \emph{error-free} execution in these experiments.

\noindent \textbf{Iago protection for FS and synchronization APIs in SGX enclaves (\S\ref{sec:models})}.
We use the models to protect real applications running in Intel SGX enclaves. We generate \BackendVerify and integrate it with Graphene-SGX~\cite{graphenesgx:usenix17}. 
\BackendVerify intercepts \emph{all} calls to the respective OS services, and 
validates conformance to the models. 
We execute two large real-world applications: Memcached~\cite{memcached}, and SQLite~\cite{sqlite},
and observe the added assertions has negligible performance overheads.

\noindent \textbf{Targeted vulnerability checker (\S\ref{sec:eval:fuzzing})}. We test for Iago vulnerabilities in existing applications and libOSs by generating \BackendFuzz from the FS model.
We find one new vulnerability in Graphene-SGX, which we responsibly disclosed (\BackendVerify successfully detects this vulnerability). Additionally, we test large applications such as Redis~\cite{redis} and detect
more vulnerabilities than the Emilia fuzzer~\cite{emilia:ndss21}. 

%% file: background.tex
\section{Background and threat model}
\label{sec:background}

\paragraph{Trusted Execution Environments (TEEs)}
TEEs offer an isolated execution environment protected from privileged software adversaries.
Two types of TEEs exist: \emph{secure enclaves} that protect a part of processes' address space~\cite{sgx:hasp13:mckeen, amazon:nitro-enclaves, keystone:eurosys20, sanctum:usenix_sec16},
and \emph{trusted virtual machines} (VMs) such as AMD SEV,  Intel TDX, and ARM CCA~\cite{amd_sev_snp:white_paper,tdx:white_paper, arm_cca} that protect entire VMs with their OS.
The hardware protects a TEE context's control-flow integrity and manages TEE's entry, exit, and exceptions, thereby shielding the TEE from a strong adversary.

\paragraph{Enclaves}
Hardware-based enclaves are supported in many cloud platforms~\cite{markruss:azureconfidentialcomputing:2017,ibm:sgx:blog,alibabacloud:ebm:sgx}.
Moreover, even trusted-VM TEEs are also used as an enclave to reduce their TCB~\cite{enarx:github}. As Intel SGX enclaves are the most mature technology, we use them as the target platform for prototyping our work.

\paragraph{Untrusted interface attacks} 
\citeauthor{iago:asplos13}~\cite{iago:asplos13} 
demonstrated that an untrusted OS can perform attacks that may break
the application's control flow integrity. They coined the term \emph{Iago attacks}.  In brief,
Iago attacks return maliciously crafted values instead of valid results. 
Later, \citeauthor{tale_of_two_worlds:ccs19}~\cite{tale_of_two_worlds:ccs19}
generalized Iago attacks to calls to an untrusted interface from SGX enclaves. In this paper, we generalize this notion further by referring to any TEE (not just SGX) and permitting the adversary to perform arbitrary  modifications to the service behavior. 

\noindent \textbf{TEE runtimes.}
One common trait among enclave technologies is the difficulty in executing legacy applications because they cannot directly invoke untrusted external functions, such as system calls. Further, some 
OS services require in-enclave runtime support (e.g., {\tt fork()} in SGX). LibOSs alleviate this problem~\cite{scone:osdi16,graphenesgx:usenix17,haven:osdi14,sgxlkl, panoply:ndss17}. 
They place the entire application and its library dependencies in the enclave 
and serve system calls by implementing them internally or forwarding them to the OS, in which case they may be targeted  by Iago attacks. 
LibOSs usually protect I/O-related system calls using encryption and integrity authentication tags. This allows applications to use external OS services, such as a host FS (e.g., in Keystone~\cite{keystone:eurosys20}, Komodo~\cite{komodo:sosp17}, and Amazon Nitro~\cite{amazon:nitro-enclaves}) with the data protection implemented by the libOS. 
SDK-based enclave applications aim to minimize the TCB by breaking existing applications 
into trusted and untrusted components. This allows control 
over the untrusted interface implementation, which includes the protection mechanism employed.
Finally, in trusted virtual machines the guest VM can secure the data plane similarly
from an untrusted hypervisor.
We collectively refer to these frameworks facilitating TEE usage and protection as TEE runtime.

\paragraph{Threat model}
We focus on the standard TEE threat model~\cite{sgx:hasp13:mckeen,amd_sev_snp:white_paper}, which excludes
side-channel, speculative execution, and denial-of-service (DoS) attacks.
In the context of enclaves we do not consider attacks on the enclave ABI tier, which may exploit incorrect
register values and incorrect sanitization by the TEE runtime, or attacks due to bugs in the trusted code, such as buffer overflows. These attacks are orthogonal and their mitigation is well-understood ~\cite{memory_sok:ieeesp13, tale_of_two_worlds:ccs19}.

%% file: motivation.tex
\section{Motivation}
\label{sec:motivation}
TEEs may define different trust boundaries: in enclaves, the OS services are out of the Trusted Computing Base (TCB). In VM-based TEEs the OS is trusted, but the hypervisor services are not. However, many programs rely on services, which are no longer trusted to function correctly. 
For example, to access a host FS from enclaves, or a virtio device from a trusted VM, programs invoke respective system- or hyper-calls. 

Excluding these services from the trust boundary is supposed to improve security by reducing the TCB. However, accessing untrusted services without special care might expose the trusted software to Iago attacks.

Using untrusted services is not unique to TEEs, and was considered in other contexts. For example, programs executing on top of microkernels~\cite{l4re} may also require access to an untrusted FS to run legacy software~\cite{vpfs:eurosys08}. However, we are not aware of systematic solutions to this problem so far.

\subsection{Examples of Iago attacks}
\label{sec:motivation:attacks}

Many examples of Iago attacks have been published in prior work~\cite{besfs:usenixsec20,iago:asplos13,emilia:ndss21}.
Our goal in this section is to show that mitigating them is not trivial.
As an example, we use a compromised FS that attacks enclave applications using it.

\paragraph{Handled by libOS: data tampering} Malicious OS returns incorrect file contents. This attack is easy to mitigate, and most existing advanced libOSs  do so by using well-known secure cryptographic integrity tools~\cite{haven:osdi14, graphenesgx:usenix17, intel:sgxsdk}. 

\paragraph{API attack: incorrect return values}
A malicious FS returns an already existing file descriptor for an \id{open()} call.

\begin{lstlisting}[style=CStyle]
fd1=open("foo1", O_CREAT | O_RDWR, 0644);
// OS returns fd1 value maliciously
fd2=open("foo2", O_CREAT | O_RDWR, 0644); 
// TEE runtime updates auth tag for fd1 at offset 0
write(fd1, w1_buf, 100); 
// TEE runtime updates auth tag for fd1 at offset 100
write(fd2, w1_buf, 100); 
\end{lstlisting}

This is a real attack that exploits a vulnerability in the Graphene-SGX FS protection layer that we find automatically using \BackendFuzz.
The vulnerability stems from the way Graphene-SGX handled data tampering attacks above. Specifically, it maintains a shadow in-TEE state for each opened file descriptor to store data authentication tags for the file contents at block granularity, and updates them on every write. Since the protection layer uses the OS-returned file descriptor value to index the shadow state it results in incorrect storage of the authentication tag if the file descriptor is incorrect. 

\paragraph{API attack: incorrect behavior} Malicious OS creates a new file instead of a link to an existing file. Accessing both the linked and original files would result in inconsistent content.

Data integrity validation is insufficient to protect against API attacks since
the file contents are returned correctly. The attack on the {\tt link} call is particularly difficult to mitigate without validating the file system state is updated \emph{correctly}.

\subsection{Manual protection}
\label{sec:motivation:manual_protection}
Existing mitigation approaches rely on correctness validation 
following the trust-but-verify model~\cite{inktag:asplos13,sego:asplos16}.
In a nutshell, TEEs invoke untrusted interface calls and validate that the response matches the expected service semantics.

\paragraph{Interface complexity}
The difficulty of establishing a secure perimeter for an application and 
supporting access to untrusted services depends on the service semantics complexity.
Taking SGX enclaves and libOSs as an example,
there are several approaches for providing secure access to untrusted OS services, which trade the size of the TCB with the protection complexity. We explain this tradeoff next, using an FS API as a running example.

\paragraph{Complex semantics, smaller TCB}
LibOSs, such as Graphene-SGX~\cite{graphenesgx:usenix17} and
SCONE~\cite{scone:osdi16} forward system calls to an untrusted FS.
In turn, they internally implement a comprehensive 
FS layer that strives to protect against interface attacks. 
Unfortunately, this protective layer is written manually and the effectiveness of the protection is hard to validate.
Our work discovered a vulnerability in the Graphene-SGX protection layer as we mentioned in the example above.

\paragraph{Simple semantics, larger TCB}
Some other libOSs include a partial or complete FS
implementation~\cite{sgxlkl, haven:osdi14, occlum:asplos20}, 
reducing reliance on the untrusted OS.
For example, Haven~\cite{haven:osdi14} and SGX-LKL~\cite{sgxlkl} 
include a complete FS implementation in the libOS, 
which reduces the interface to virtio-blk with simpler semantics and fewer inter-dependencies between operations.  
Unfortunately, a narrower interface still requires a protection layer of its own.
This is  written manually by libOS developers and shares the same validation problem as before.

Furthermore, the inclusion of an FS implementation in the TCB has multiple disadvantages. A larger TCB implies a larger probability of bugs. Further, enclave applications are forced to use a particular FS offered by the libOS, 
instead of any host OS-supported FS.
Moreover, the deployment of an internal FS complicates integration with the host: users can use neither the existing FS structure 
nor tools such as backup with the \id{rsync} utility. 
Last, the libOS implementation replaces a mature and continuously maintained FS in the OS,
such that security patches to it arrive late.

These examples show the insufficiency of manual protection approaches. 

\begin{figure}[t]
    \centering
	\includegraphics[width=.80\columnwidth]{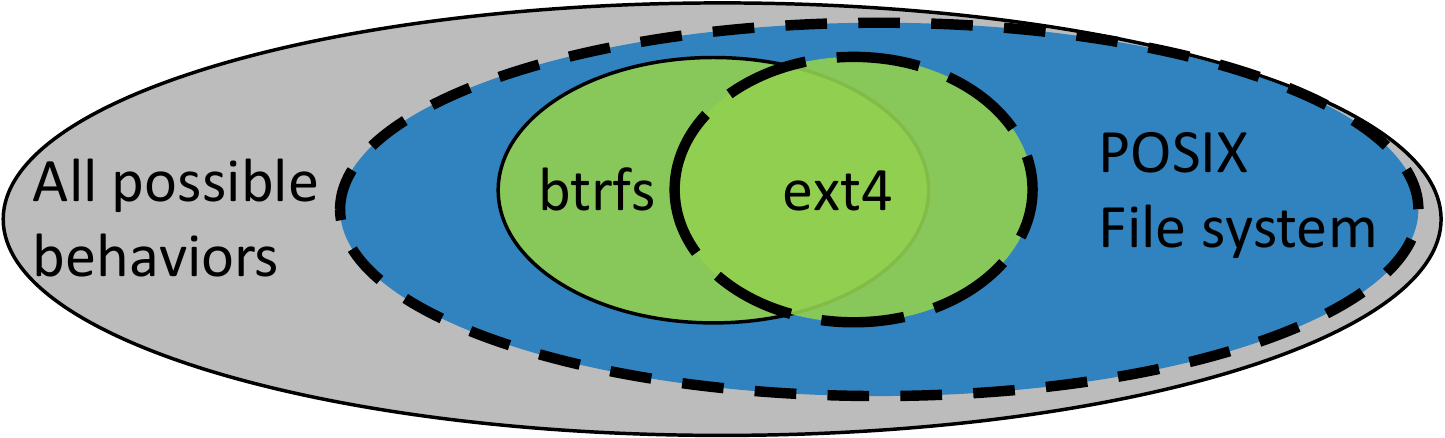}
	\caption{Two possible FS models to protect against Iago attacks (dashed outline): one tighten to the  POSIX spec ( multiple valid behaviors) and another to the ext4 implementation.}
	\label{fig:model_correctness}
\end{figure}

\subsection{Model-based protection to the rescue}
\label{sec:motivation:model_protection}
Recent successes in applying formal methods to developing provably-correct operating systems and services~\cite{sel4:sosp09,besfs:usenixsec20,yggdrasil:osdi16,hyperkernel:sosp17,lightweight_formal:sosp21} motivate us to consider a principled approach to mitigation of Iago attacks that a service functional correctness model can drive.

However, we found it challenging to apply these approaches to achieve our goals. 
The push-button verification methods used to develop an OS~\cite{hyperkernel:sosp17} and an FS~\cite{yggdrasil:osdi16} had to modify certain aspects of the modeled services (i.e., via finitization) to allow proof completion. In contrast, we aim to model existing unmodified services.

BesFS~\cite{besfs:usenixsec20} developed a Coq-based model to prove correctness invariants but required manual translation into C to integrate the FS implementation into an SGX enclave.

Fundamentally, these approaches require extensive expertise in formal methods to develop and prove, whereas our goal is to make model development accessible to non-experts.

We share this sentiment with the recent work on the application of \emph{lightweight} formal methods to check a complex application~\cite{lightweight_formal:sosp21}. They develop a simple reference executable model (mock) of a service to check compliance with a full-service implementation ~\cite{lightweight_formal:sosp21} by running them side-by-side. However, we cannot apply this approach as is. First, it checks for exact equivalence of outputs with the reference model, and thus cannot be used to represent services that can produce multiple valid results (i.e., different FSs that conform to POSIX specification, or file {\tt read} that returns a different number of bytes). In addition, it is not clear how to check the blocking \emph{behavior} of calls such as \id{mutex\_lock()} with the help of the reference model. Last, this approach does not allow us to leverage the model to generate a vulnerability checker. 

\paragraph{Challenge: correct model}
To be useful for mitigating Iago attacks, a model should correctly represent a concrete reference implementation or  high-level service specification (see Figure~\ref{fig:model_correctness}
for illustration). However, proving model correctness is an open problem~\cite{spec_guard:oospla17}. An attacker can leverage a model bug to compromise the system. 
Formally proving key model correctness properties  indeed increases the confidence in the model correctness, yet is challenging for non-experts.

In \projname, we seek to develop a complementary approach to create service \specifications, validate they faithfully reflect the correct service implementation using existing tests, and reuse them to generate both a validator and a vulnerability checker to mitigate Iago attacks systematically. 

%% file: overview.tex
\section{Approach overview}
\label{sec:overview}
We use two model examples to demonstrate our approach: a mutex, and a simplified FS with a single directory that is empty at the start, files have no attributes other than their size and permissions, and there are no links. We  use a simplified {\tt read()} call to demonstrate the main concepts.  

\subsection{Service model}
\label{sec:overview:model_dev}
In \projname, developers write models in the \TeeSpec DSL. 
Each model encapsulates the \emph{user-visible state} of the untrusted interface and a set of operations that manipulate this  state to reflect
the expected behavior of the modeled interface.  

\paragraph{FS abstract state}
In our simplified FS, the state includes three \emph{abstract maps}. 
The first map holds the file size and the contents\footnote{For simplicity, the file data is stored in the model. In practice, we keep only data authentication tags}  indexed by a unique identifier (ino) akin to an inode. 
The other two maps translate paths or open file descriptors to ino values and the cursor.

\begin{lstlisting}[style = TeeSpecStyle]
Map ino_state(ino: int) returns(sz:int, data:char[])
Map fs_state(path:string) returns(ino:int); 
Map fd_state(fd:int) returns(off:off_t, ino:int); 
\end{lstlisting}

\paragraph{FS operations}
For each of the FS API calls, the model includes a respective \emph{action}. An action resembles regular program code and specifies the state updates performed by the call. 
The action specifies constraints  on the call's semantics  using the abstract states via \id{requires} statements. This is an intuitive way to specify conformance of the untrusted call to the model. The {\tt read()} action sketch is as follows:

\begin{lstlisting}[style = TeeSpecStyle,caption={{\tt read} call model.},captionpos=b, label={lst:read_spec}]

action read(fd: int, buf: void[], cnt: size_t) 
       returns (nread: ssize_t) := {
 # valid file descriptor?
 if (fd_state(fd) == NULL) return -EBADF;
 # untrusted interface call
 nread := extern call untrusted_os_read(fd,buf,cnt);
 requires (nread >= 0 and nread <= cnt); 
 off:off_t := fd_state(fd).off;
 ino:int := fd_state(fd).ino;
 # cannot read past end of file
 requires ((cnt >= ino_state(ino).sz - off) -> 
  nread <= (ino_state(ino).sz - off));
 requires (ino_state(ino).data[off:off+nread] ==
  buf[0:nread]);
 # update cursor
 fd_state(fd).off := fd_state(fd).off + nread;
 return nread;
}
\end{lstlisting}

A \lstinline!return! before the untrusted interface call is a shorthand
to invoking the call and checking that it returned this value.
However, this structure is more succinct, as it does not require
checking for individual errors and allows specifying preconditions
along with corresponding error codes for when they are violated.

\paragraph{Mutex model}
Modeling mutex is different since operational behavior 
must be specified additionally to the return value.
The abstract mutex state is represented by an abstract map holding
a counter representing the mutex is locked (>0) or unlocked (0)
indexed by a mutex identifier. 

\begin{lstlisting}[style = TeeSpecStyle]
Map mutex_state(id:int) returns(counter:int);
\end{lstlisting}

The action corresponding to {\tt mutex\_lock} is as follows:

\begin{lstlisting}[style = TeeSpecStyle, caption={{\tt mutex\_lock} call model.},captionpos=b, label={lst:mutex_spec}]
action mutex_lock(id: int) returns (res: void) := {
 extern call untrusted_os_lock(id);
 atomic (mutex_state(id)) { 
  await requires (mutex_state(id).counter == 0);
  mutex_state(id).counter:=mutex_state(id).counter+1;
 };}
\end{lstlisting}

It invokes the untrusted mutex call first, and then \emph{atomically} checks (thanks to the keyword {\tt atomic})  that the respective mutex is indeed unlocked according to the model state, and turns it into locked. Indeed, the model can successfully identify the case when the untrusted mutex call would not block even though the mutex is locked. The \emph{await} keyword is used to allow correct generation of \emph{blocking calls} by  
\BackendMock (\S\ref{sec:design:mock}). We explain correctness in \S\ref{sec:models:sync_primitives}.
 
\paragraph{Summary: models}
While these examples are simple they demonstrate the core concepts of the models used in \projname.
First, models are free of implementation details (i.e., the waiting queue for mutex). 
Second, the models are both succinct and expressive. Modeling \id{mutex\_lock()} 
only requires a few lines while capturing important semantic behaviors of the call: 
it is a blocking call, and abstract state checks and updates must happen atomically with respect to other operations performed on the same abstract mutex object.

\begin{figure}[t]
    \centering
	\includegraphics[width=.95\columnwidth]{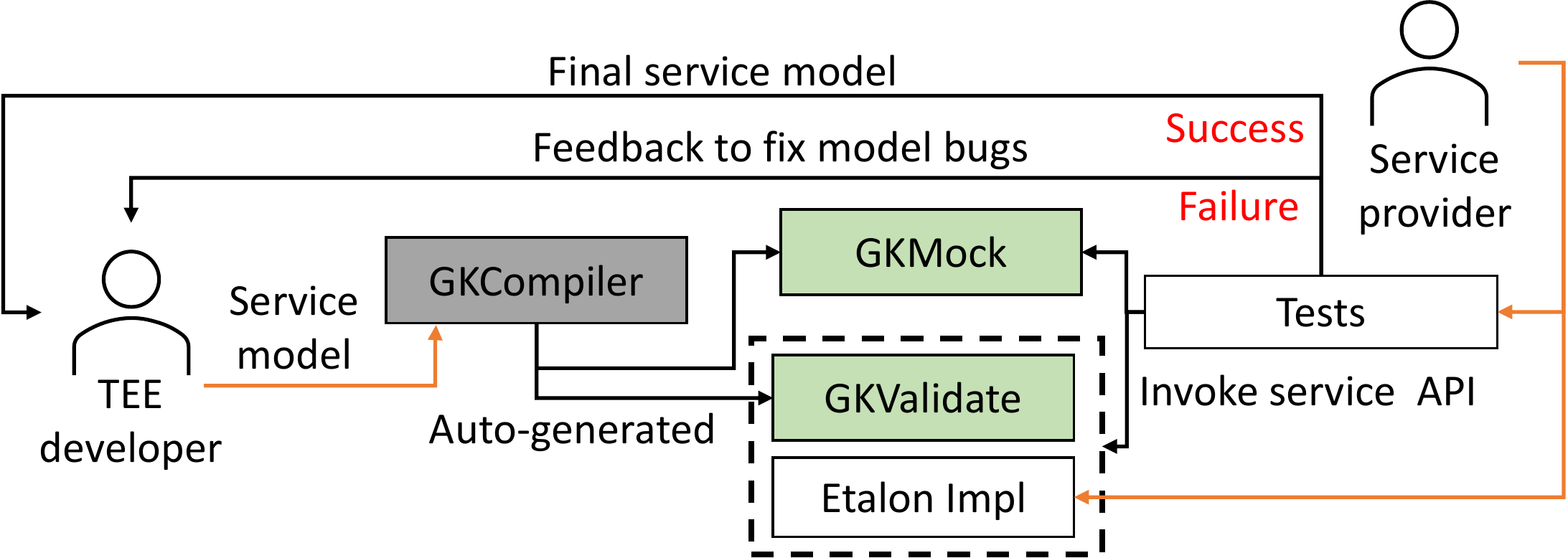}
	\caption{Model development flow.}
	\label{fig:prot_debug}
\end{figure}

\subsection{From model to \BackendVerify}
The model cannot be used as is to protect TEEs. 
TEE developers use \projname's compiler to auto-generate
a \emph{correct-by-construction (w.r.t. the model)} \BackendVerify module. \BackendVerify is a C library that integrates  with TEE runtimes to validate the service's conformance to the model (see Figure~\ref{fig:prot_flow}, left). 
\BackendVerify tracks \emph{concrete values} used by TEE programs, invokes untrusted calls, and performs runtime assertions based on the model's constraints in \id{requires} statements.

Consider the example in-enclave function in Listing~\ref{lst:conf_test} (we also use it in another context, hence the name). \id{create}, \id{write} and \id{read} functions invoke the \BackendVerify calls which internally call and validate the respective system calls. In this example, the model state is updated as follows; the {\tt fs\_state} map will contain an entry ``foo'' with a unique \id{ino}. \id{ino\_state} map in \id{ino} will store the first $\mathit{nw}$ bytes from the \id{w\_buf} as the file's content. After the untrusted system call invocation, runtime asserts are invoked, e.g.,  to validate that the file contents returned
by the untrusted call match the expected contents in \id{inode\_state} map.

\begin{lstlisting}[style=CStyle,emph={conformance_read},emphstyle=\color{mBlue},caption={Sample conformance test.},captionpos=b, label={lst:conf_test}]
void conformance_read() {
 fd = create("foo", S_IRUSR | S_IWUSR);
 assert (fd >= 0);
 char w_buf[100], r_buf[100];
 memset(w_buf, 0xff, 100);
 int nw = write(fd, w_buf, 100);
 lseek(fd, 0, SEEK_SET);
 int nr = read(fd, r_buf, nw);
 assert(!memcmp(r_buf, w_buf, nr));
}
\end{lstlisting}

\paragraph{Model initialization}
In our example, we consider an empty FS with no files. 
But if there were (i.e., ``bar'' in the listing below), \projname allows specifying initial concrete state that would be compiled into \BackendVerify, linked with the TEE
runtime, and included in its attestation report~\cite{costan:sgxexplained:cryptoeprint16}. 

\begin{lstlisting}[style = TeeSpecStyle]
# Example: FS contains a single empty file "bar"
init { 
 fs_state("bar").ino := 0;
 ino_state(0).sz := 0; 
}
\end{lstlisting}

\begin{figure}[t]
    \centering
	\includegraphics[width=.95\columnwidth]{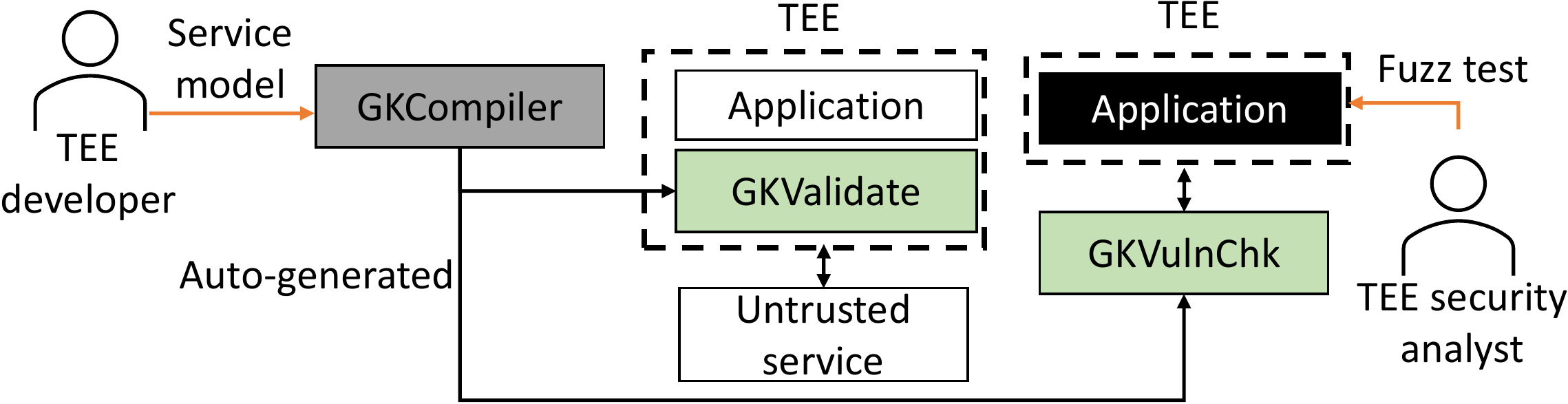}
	\caption{Model-driven TEE protection and fuzzing in the production stage.}
	\label{fig:prot_flow}
\end{figure}

\subsection{Model debugging and tightening}
\label{sec:overview:model_debug}
\projname facilitates iterative model debugging and tightening. A model developer runs the tests used in the original service development, fixes the model based on the outcome, and reruns the tests until an error-free test execution (Fig~\ref{fig:prot_debug}).

To enable this process, \projname compiler uses the model to generate \BackendMock, which is a functional executable implementation of the service (also called a reference implementation in~\cite{lightweight_formal:sosp21}).  \BackendMock for the FS model is a functional in-memory FS.
Unlike the validator, the mock does not invoke the real service. Instead, it encodes the \id{requires} constraints to
a Satisfiability Modulo Theories (SMT) formula, and uses an SMT solver to compute the return value that satisfies it. It chooses a random value if multiple are permitted, or generates an error when none is found.

\projname model debugging approach relies on service test-suites. The tests are invoked both on the etalon implementation through the \BackendVerify, and on the mock (Figure~\ref{fig:prot_debug}).
This process strives to achieve a tight approximation of the model to the service (assuming high coverage of the available service tests), which is highlighted in the following Lemmas.

\begin{lemma}
If a test invoked on the \BackendVerify-on-etalon combination 
fails then the model is over-restrictive. 
\label{lemma:validator}
\end{lemma}

\begin{proof}
Assume the model is not over-restrictive, there are no added runtime assertions that would cause the test to fail. 
Therefore, the test should pass. 
\end{proof}

For example, assume that in the \id{read()} model (Listing~\ref{lst:read_spec}), we replace the constraint {\tt nread <= cnt} with {\tt nread < cnt}.
As a result, the test in Listing~\ref{lst:conf_test} would fail.

\begin{lemma}
If a test invoked on a non-over-restrictive \BackendMock fails, then the model is over-permissive.
\end{lemma}

\begin{proof}
Assume the model is not over-permissive and not over-restrictive, 
the mock-returned values must satisfy all the available constraints.
A conformance test should not have an invalid assertion and the test should pass. 
\end{proof}

For example, assume that we drop the constraint on the number of bytes read from the \id{read()} model. 
Thus, the test in Listing~\ref{lst:conf_test} would fail when the 
mock would return $1000$ bytes, which satisfies a vacuous constraint.

When running tests on \BackendMock we execute each test multiple times,
to test the model behavior when multiple possible outputs are permitted.

\subsection{From model to \BackendFuzz}
\label{sec:overview:fuzz}
We use a model to generate \BackendFuzz to enable targeted Iago-vulnerability checking of TEE programs that cannot be modified.  \BackendFuzz is a C library to be integrated into a fuzzer (see Figure~\ref{fig:prot_flow}, right).
Internally, \BackendFuzz generates a set of 
malicious values that \emph{violate} the model constraints 
\emph{tailored to the tested program}.

For example, consider testing the program in Listing~\ref{lst:conf_test}. 
When invoked with \BackendFuzz instead of real service, the {\tt read} call  may return incorrect file contents or an incorrect number of bytes read.
This approach allows targeted and more effective fuzzing. 
However, we allow minimizing the search space even further with the model 
hints (\S\ref{sec:design:fuzzer}). 

%% file: design.tex
\section{Design}
\label{sec:design}
\paragraph{Model language}
\label{sec:design:lang}
We settle on representing service in an abstract model rather than developing a simplified reference implementation~\cite{lightweight_formal:sosp21} to simplify the derivation of \BackendVerify,\BackendVerify and \BackendMock. 
Prior work used either language intended for modeling such as Alloy~\cite{alloy,P:msr,spin:ieeese}, or a subset of general-purpose languages such as Python and Rust to write models as they are used for the implementation as well~\cite{yggdrasil:osdi16, hyperkernel:sosp17, lightweight_formal:sosp21}. We introduce a custom DSL: \TeeSpec, with a syntax 
that is close to C as seen in examples in Section~\ref{sec:overview}. This is a pragmatic choice, 
as we find it easier to implement a compiler for \TeeSpec instead of a Python compiler.

\paragraph{\TeeSpec details}
Types in \TeeSpec closely follow the C language types. 
We remove pointers, however. Instead, arrays and nested
arrays are supported as in C with explicit element type. 
We add the \id{string} type to distinguish strings from character arrays.
The operators resemble C, except for small changes such as range comparison and assignment for array types.
In addition, \TeeSpec introduces 
quantifier logic operators within the \id{requires} statements, 
abstract maps, \id{atomic}, \id{await}, \id{fuzz}, and 
external calls into C functions.
They are parsed by the compiler to generate all the artifacts as explained next.

\paragraph{Model compiler}
The compiler generates C code from input models.
This choice is due to three reasons. 
First, it enables simple
integration with existing systems, as
we show with Graphene-SGX (\S\ref{sec:evaluation}),
FUSE~\cite{fuse} for mounting a mock FS, and 
existing test suites, e.g., the Linux Test Project (LTP)~\cite{ltp}. 
Second, TEE developers can use familiar tools such as profilers, debuggers, and compilers with mature optimizations, which reduces 
the validator's performance impact.
Third, \TeeSpec similarity to C simplifies the compiler.
In addition, \projname includes a trusted runtime library with 
data structures and functions invoked by the compiler as described next.

\paragraph{Code generation}
Generating \BackendVerify, \BackendMock, and \BackendFuzz from a model
has common building blocks.
The compiler parses the model, assigns a type to each variable, and checks for type safety. Code generation for C-like statements 
is straightforward. We focus on the new DSL statements next.
First, abstract maps are generated using hashtables: 
the compiler emits code for memory management such that
entries are allocated before inserting them into the hash table and deleted when they are freed.
Assignments and validations using abstract maps in the model are 
translated to GET and SET operation on the hash table entry. 
Primitive types assignment is handled with the C assignment operator.
Complex types rely on memory copying functions, 
e.g, buffer assignment via \id{memmove}.
External calls are a simple invocation of existing functions, and
the \id{init} keyword generates a function with unique name 
for the concrete \specification initial state. 
This function is intended to be invoked at the application's startup.

\paragraph{Atomicity}
To support fine-grained atomic transactions via the \id{atomic} keyword the compiler
emits a spinlock in each hash table definition.
Spinlocks are values stored in memory and used
to implement user-space spinlocks that do not rely on untrusted OSs.
These spinlocks cannot be used in user programs instead of OS futexes or mutexes, however. 
The reason is that they lack much of the functionality of futexes, 
including fairness. Instead, they are intended to protect only a few memory assignments in the model itself.

\paragraph{Correct-by-construction}
Assuming the compiler correctly parses the model
and has no bugs, the 
executables are correct-by-construction w.r.t the model.
Specifically, \BackendVerify asserts the correct semantics of the untrusted service;
\BackendMock emulates the untrusted service's semantics; \BackendFuzz returns only values that violate the untrusted service's semantics.

\paragraph{\BackendVerify}
For \BackendVerify, the compiler generates a runtime assertion per \id{requires} statement.
The predicate logic expressions are generated via the respected C operators.
Notably, universal and existential quantifiers are generated using loops. Similarly, higher-order logic 
generates nested loops.

\paragraph{\BackendMock}
\label{sec:design:mock}
For \BackendMock,  
the compiler aggregates scoped constraints for every variable,
encodes them to an SMT formula and invokes z3 SMT solver~\cite{z3}.
The solver outputs values that must satisfy all the constraints, which in turn 
are assigned to the corresponding variables.

The services behavior is not limited to updating values.
For example, an attempt to lock an already locked mutex  
cannot return until the same corresponding mutex is unlocked.
The keyword \id{await} marks blocking calls in the model
and expresses the constraints for which they 
wait (c.f. \id{mutex\_lock()} in~\S\ref{sec:overview}). 
To simulate it, the compiler generates busy-wait loops waiting for the condition to become true.
In the case of \id{mutex\_lock()}, \BackendMock waits for another thread to unlock: 
change the lock state, which clears the condition.

\paragraph{\BackendFuzz}
\label{sec:design:fuzzer}
\BackendFuzz is generated similarly to \BackendMock, except 
the constraints are negated before passing them to the SMT solver.
The compiler generates a configurable number of solutions by calling the SMT solver
with added constraints for the previously generated value at subsequent calls.
Unfortunately, this approach can still result in a large number 
of candidate malicious values. 
For example, the negation of the \id{read} constraints (Listing~\ref{lst:read_spec}) is:
\begin{lstlisting}[style = TeeSpecStyle]
nread <0 or nread > cnt or 
cnt>=ino_state(ino).sz -> nread>ino_state(ino).sz-off
\end{lstlisting}

We overcome this by introducing hints to \TeeSpec 
that direct the value generation in \BackendVerify to those 
that are more likely to trigger bugs. These are only soft hints, and they do not preclude the generation of all permitted values.
For example, the \id{read} call in the FS \specification chooses from valid Linux error codes first, 
or prefers large numbers that may result in buffer overflows if used to access memory.

\begin{lstlisting}[style = TeeSpecStyle]
fuzz { requires ( nread>=-131 or nread > (1<<30)); }
\end{lstlisting}

\paragraph{Model debugging and refinement}
\label{sec:design:spec_validation}
The compiler may be configured to output a trace file of invoked untrusted calls and their parameters.
We find this useful when running the validator and model implementation with the test suites.
Effectively, the trace and the failed test source code acts as a counter-example, similar to
a verifier based on SMT solvers, allowing 
to debug invalid constraints in the model.

\paragraph{Implementation}
The compiler prototype is written in PLY v3.10 
and consists of 2,770 LOC. 
\BackendFuzz and \BackendMock use the z3 v4.8.10 SMT solver~\cite{z3}.
The FS and synchronization primitives \specifications consist of 1,226
and 300 lines, respectively. 
Our trusted library consists of 2,177 lines, including the uthash~\cite{uthash} 
hash table.

\section{Securing untrusted interfaces in TEEs}

Applying our approach in TEEs poses two additional requirements: trusted initial state, and coordinated state changes. At the same time, not all external services used in TEEs can and should be protected. We explain these below. 

\paragraph{Trusted initial state}
When a TEE is started, a service usually already has some initial state that must be reflected in the model state as well. This initial state however must be trusted. 

Usually, TEEs are invoked with an initial state that is embedded with the trusted binary and thus attested at TEE invocation time. For example, Graphene-SGX uses a manifest containing trusted files and their expected paths on the host.

\TeeSpec includes special constructs for specifying the initial state. We require the TEE developer to generate these constructs from the TEE's trusted state, and then compile them with \BackendVerify, as part of the integration of the validator with the TEE runtime. 

\paragraph{Coordinated model state changes}
An internal state of an untrusted service may be modified as a result of interaction with other entities, or due to the internal service logic. 
For example, an FS may service multiple processes concurrently.

This scenario would trigger an assertion by \BackendVerify because the model state would diverge from the actual service state. While seems to be a problem, this is \emph{exactly} the behavior expected of the correct trusted system if such  service state changes were \emph{not coordinated} with the TEE. Indeed, this situation is indistinguishable from an Iago attack. 

Thus, to support services whose state can be modified outside the current TEE, the changes must be coordinated explicitly with \BackendVerify. This coordination request may originate from a trusted party, or otherwise, the request should be rejected. For example, to share a file among two enclaves they must establish a trusted communication channel (i.e, via TLS connection to a remote enclave~\cite{tls:rfc}, or via local attestation~\cite{costan:sgxexplained:cryptoeprint16}), through which they can coordinate file accesses. 

\paragraph{Trusted hardware services}
TEE technologies offer a variety of trusted hardware primitives, 
which therefore do not require additional protection.
For example, if TEE hardware provides a trusted page table, then the OS services using it would not be vulnerable to Iago attacks from the OS virtual memory subsystem. In the case of Intel SGX, for example,  process creation and virtual memory management are two cases that are validated by SGX hardware. \projname is flexible to model and validate them, but it is not necessary.

\paragraph{\BackendVerify security guarantees}
\BackendVerify protects against Iago attacks that may modify the behavior of calls to untrusted services invoked from trusted code~\cite{tale_of_two_worlds:ccs19}.
Specifically, it protects against illegal control-flow modifications of a bug-free, trusted code that can circumvent its integrity or confidentiality. 
If the service behavior is modified by the adversary but it conforms to the permitted behavior of correct service, an application must handle such cases correctly; otherwise, it is considered an application bug.

In practice, \projname cannot protect against corruption of trusted data if an  attack does not compromise a \emph{running application}. For example, a malicious OS may change the behavior of a \id{write()} call such that it  corrupts the file contents; \projname would not be able to detect this change unless the file contents are later read. In other words, \projname may not be able to catch the invocation of the maliciously modified service itself. However, it guarantees that it will prevent the attempts to use the results of such a call by TEE software.

\paragraph{Services that cannot be protected}
\projname cannot be used to secure access to \emph{externally modified untrusted state}, e.g., time sources, because the change to the state is not coordinated. 
Further, both hardware and privileged software may deny the service altogether. For example, the OS scheduler cannot be validated for SGX because the OS is in control of scheduling and may deny cycles to enclaves. The same is true for a storage drive that does not persist data. These DoS are out of scope.

\section{Experience with protecting SGX enclaves}
\label{sec:models}

To evaluate the utility of \projname we seek to protect untrusted FS API, futex OS call, and pthread mutex API in SGX enclaves. These are three complex services with non-trivial interfaces and are broadly used in trusted applications. 

We develop the respective models following the proposed iterative development approach (\S\ref{sec:overview:model_debug}) and generate their respective 
\BackendVerify, and \BackendFuzz. We integrate \BackendVerify into Graphene-SGX to protect
enclave applications, and use \BackendFuzz to find Iago vulnerabilities in Graphene-SGX and production applications.

\subsection{File system}
We model a rich FS with support for
directories, hard and soft links, data, metadata, and permissions.

\paragraph{\Specification states}
We use four groups of abstract maps: 
process-related, path-related, in-use files/directories handles, and inode-related. Process maps are used to model global metadata used 
by the FS operations, e.g., user and groups for checking permissions,
and the current working directory for path resolution.
Path-related maps are indexed by a \emph{canonical path} 
and maps to the corresponding inode.  
Note, in our \specification the inode does not represent stored disk blocks locations.
Instead, we refer to inode as an identifier that acts as a key to the inode abstract map, which in turn contains data and metadata on the corresponding directory, or link (depending on the file type), and metadata for files. We discuss the file data below. Our \specification
stores the data for different file types in different fields in the abstract map to simplify the model.
Finally, the in-use handles map contains the unique inode identifier.
For directories, we also maintain a list of visited entries to model \id{getdents} correctly by placing constraints that each entry is returned  exactly once to the user.

\paragraph{File content protection}
Our \specification contains a configurable encryption and authentication
mechanism for file contents by using external calls to 
formally-verified cryptographic implementations~\cite{evercrypt:ieeesp20}
and storing authentication tags in the {\tt inode\_state} map at block granularity. Writes update the tags and reads validate them.
This is similar to TEE runtimes that provide FS 
shields~\cite{scone:osdi16,graphenesgx:usenix17}
However, this check is redundant if the TEE runtime 
already validates the contents, so we can turn it off in the model. This is possible because we separately model the state for regular files, directories, and symbolic links. 

\paragraph{\Specification operations}
We model the core FS system calls that are sufficient for generating a mock that allows mounting it over FUSE~\cite{fuse} and passing a rigorous POSIX conformance suite (\S\ref{sec:eval:validation}).
The complete \specification in \TeeSpec is 1,226 lines, and contains the following FS 
operations:
\id{open}, 
\id{close}, 
\id{read}, 
\id{pread}, 
\id{write}, 
\id{pwrite}, 
\id{mkdir}, 
\id{chdir}, 
\id{getcwd}, 
\id{lseek}, 
\id{unlink}, 
\id{lstat}, 
\id{fstat}, 
\id{access}, 
\id{getdents}, 
\id{readlink}, 
\id{rmdir}, 
\id{symlink}, 
\id{rename}, 
\id{link}, 
\id{fchmod}, 
\id{chmod}, 
\id{lchown}, 
\id{ftruncate}, 
\id{truncate}.

\paragraph{Canonical path representation}
Using file paths as keys to our FS map
can be ambiguous, i.e.,  "f.txt" refers to the same file as "spec/../f.txt". 
We use canonical path representation as the key in the path abstract map.
To simplify the \specification we implement the path resolution 
function in C and add it to the trusted runtime library. 
This function retrieves the directory states from the \specification abstract maps. 

\paragraph{Limitations}
We do not model asynchronous I/O, signals, internal memory allocations, or read-only FS. We also exclude resource exhaustion (inodes, memory) and respective errors  as they constitute DoS attacks.  Finally, our FS model must run in a single thread, the limitation shared by prior works on formal modeling of FSs~\cite{besfs:usenixsec20, yggdrasil:osdi16, sibylfs:sosp15}. These limitations do not preclude running large real-world applications (\S\ref{sec:eval:validator_overhreads}).

\subsection{Synchronization primitives}
\label{sec:models:sync_primitives}
We model three types of mutexes,
\id{normal}, \id{errcheck}, and \id{recursive},
and futex \id{wait} and \id{wake} operations. 
These synchronization primitives are broadly used by multithreaded applications running in enclaves, both small (with Intel SDK~\cite{intel:sgxsdk}) and large (with a libOS). 
Graphene-SGX uses the futex API to implement higher-level synchronization primitives. 

\paragraph{\Specification states}
We use mutex and futex abstract maps with the respective identifier as a key that maps
to metadata values, e.g., counter that is used to model
a recursive mutex.

\paragraph{Mutex operations}
The {\tt mutex\_lock} implementation (Listing~\ref{lst:mutex_spec}) does not enclose the call to the untrusted mutex into the \id{atomic} clause. Thus, {\tt mutex\_lock}, and the state access is not atomic with respect to each other. This might seem to lead to state divergence and a potential TOCTOU attack.
A strawman approach to include the untrusted call into the atomic block would cause a deadlock.

We observe, however, that the atomicity of the model state check and the untrusted call invocation is unnecessary.
For \id{mutex\_lock},  the states are checked and updated after the invocation 
of the untrusted call, whereas for \id{mutex\_unlock} the checks and updates 
are performed before the call. In both cases, the \specification validates that the
untrusted call indeed succeeds. 
As a result, even if the untrusted mutex is compromised, an attempt to
lock it twice would fail on validation assertion (rather than block). 
Further, even if the unlocking thread is preempted after the 
state is updated but before the untrusted mutex unlock is invoked, the calling thread already left the critical section protected by the mutex, and will resume from the same point.

\paragraph{Futex operations}
The original futex API is as follows:  
\id{futex\_wait} puts the current thread to sleep 
and \id{futex\_wake} wakes one or several sleeping threads.
The first argument is a pointer and is used as a unique futex identifier. 
The pointed value acts as a condition to block or not. 
Thus, the pointed address must be accessible to the OS (cannot be stored 
in the TEE's memory).
A malicious OS can modify the memory contents, without the TEE being notified.
This problem precludes secure updates and validation of the futex states.  

We introduced new calls to the futex API:
\id{futex\_init}, \id{futex\_destroy}, and \id{futex\_cmpxchg}.
This allows us to maintain a shadow state in the TEE's memory and atomically
update it together with the untrusted variable used for the futex. 

The \specification tracks sleeping and waking threads to validate that the 
OS cannot the too large number of woken up threads. If the OS does not wake up these threads, it is equivalent to a DoS attack.

\paragraph{Model verification attempts}
\label{sec:overview:formal}
In an attempt to improve the model correctness guarantees, we experimented with formal verification  manually translating the model to mypyvy~\cite{mypyvy:cav19,ivy:pldi16}, which 
was previously used to prove complex transition system protocols such as Raft~\cite{raft:usenixatc14}.

As an initial step, we translated our mutex \specification to prove that it guarantees mutual exclusion.
The translation was straightforward and took a single day.
We encode the mutual exclusion property in mypyvy 
as an invariant that each mutex is held by a single or no thread
at any time. 

\begin{lstlisting}[style = TeeSpecStyle]
safety forall m in mutex_state :: mutex_state(m).counter == 1 || mutex_state(m).counter == 0;
\end{lstlisting}

These results are promising and encourage us to pursue 
verification of \projname
\specifications in the future.

\paragraph{Limitations}
We do not model timeouts (SGX has no trusted time source) or multi-process support. 
Non-default mutex types are supported but not other non-default attributes, 
such as protocol, sharing, and robustness.
For futex, we provide a partial \specification that supports the most frequently used features, 
\id{wait} and \id{wake}. 

\subsection{Model development and deployment}

\paragraph{Development effort}
The development of \projname was done in tandem with 
the FS \specification and took a roughly one-person year.
The futex and mutex \specifications took a month by an undergraduate student.

The \specification debugging approach helped discover numerous bugs,
which the student was able to fix using the feedback provided by the traces of the failed tests.
This positive experience indicates that non-experts in formal methods can efficiently develop services \specifications.

We describe two bugs discovered in an early prototype.

\paragraph{Over-restriction bug}
We placed an incorrect constraint asserting that 
the buffer size in the \id{readlink} call should always match the
size of the symbolic link in the corresponding LTS state.
However, the \id{readlink} semantics \emph{permit} such a case and
expect the partial target to be copied to the buffer. 
This was detected by running the respective tests with the validator on top of the ext4 FS.

\paragraph{Over-permission bug}
We omitted the \specification of a constraint on the returned number 
of links for the \id{fstat} system call.
When running the test suites, we noticed failures. Further analysis  revealed that the mock did not generate a 
the correct number of links because of this missing constraint.

\paragraph{Mock testing}
We use several approaches. To run conformance tests on the FS,  we connect it as a backend to FUSE and mount it regularly.
To test the mutex-mock  we use LD\_PRELOAD when running the tests. Futex tests required manual integration. 

\paragraph{Integrating validator with SGX}
We integrate the \BackendVerify into Graphene-SGX v1.1~\cite{graphenesgx:usenix17} to provide systematic protection when it accesses untrusted FS services and futexes.
We modify 274 LOC.
We place the generated validator in trusted code replacing the original Graphene-SGX logic that invokes untrusted calls. 
Finally, the validator code and the initial model state are compiled into a single enclave executable and attested together.

\paragraph{Vulnerability checker}
We use two approaches to test for vulnerabilities using \BackendFuzz.
First, to fuzz existing applications we use \BackendFuzz with Emilia~\cite{emilia:ndss21},
a fuzzer that intercepts system calls and detects crashes due to invalid memory accesses. \BackendFuzz augments Emilia's original value generator.
Second, to fuzz the Graphene-SGX protection layer we replace system call invocation
made by the trusted code with calls that return maliciously generated values. 
We do not use Emilia in this case, because its system call interception mechanism  is based on \id{strace}, making the fuzzing too slow, and because it may detect Iago bugs in untrusted code that are unrelated to the TEE.

%% file: evaluation.tex
\section{Evaluation}
\label{sec:evaluation}

\paragraph{Setup}
\label{eval:env_setup}
We evaluate \projname on 
a server with Intel Skylake i7-6700 
4-core CPU with 8~MB LLC, 16 GB RAM,
Ubuntu Linux 18.04 64-bit, Linux kernel v4.15.0-135
and Intel SGX driver v2.10~\cite{intel:sgx_driver}. 
We run each experiment 10 times and report the mean execution time.
The standard deviation is below 5\%.

\subsection{\Specification validation}
\label{sec:eval:validation}
For all the tests we compile the \BackendMock and \BackendVerify
with address sanitizers~\cite{address_sanitizer:atc12} 
to ensure that the compiler-generated code
is free of memory vulnerabilities.

For the FS model, we use the SibylFS POSIX conformance test suite, with 21,068 tests that achieve 98\% coverage of SibylFS's POSIX model~\cite{sibylfs:sosp15}. For synchronization primitives, we use LTP~\cite{ltp} stress tests performing millions of mutex operations with 120 threads. 
We also test for conformance using mutex and futex LTP 
tests excluding tests with unsupported features by our \specification (8 out of 14). 

\paragraph{File system}
The SibylFS test suite is designed to work against a real FS. 
We use FUSE~\cite{fuse} for that purpose.

Testing with FUSE poses a challenge. 
The tests use
multiple processes and assume a shared FS,
while our \specification is intended for a single process.
To allow testing with multiple processes, we add an action to our FS model that updates the current process properties (user, groups, cwd, and umask). 
The mounted FS invokes this function to update the requesting process's properties
before issuing any FS operation.

For the validator-FS and mock-fs, all tests terminate successfully, implying that 
To validate compliance, we also analyze the traces
using the SibylFS compliance checker, which rejected 528 out of 21,068 tests.
Manual investigation revealed that none were rejected because of bugs in the 
\specification itself. Two traces were rejected, as they exercised sparse files
not included in our \specification. In another trace, 
the FS reported an empty root folder 
with two links instead of one. However, this is a common behavior, also found in
ext4. 
The rest of the rejected traces revealed bugs in FUSE, 
which did not forward the requests to our
handlers. 
These bugs hide potential bugs in the \specification, but we did not
investigate this further. 

Finally, we found that the same traces are rejected when mounting a
pass-through FS on top of ext4, which increases our confidence that our 
FS \specification is tight to ext4. 

\paragraph{Synchronization primitives}
We run the mutex stress-test from LTP~\cite{ltp}, which generates
high contention on locks, assuring that mutual exclusion is always maintained.
To test our mutex \specification, we intercept mutex operations using \id{LD\_PRELOAD} 
and replace them with calls to the mutex's validator and mock.
To test the futex \specification, we integrate the futex's validator and mock
into Graphene-SGX's mutex implementation, replacing the original futex operations.
Next, we run five mutex and three futex conformance tests.
We select only the tests that use features in our \specification,
e.g., without timeouts.
All tests terminate successfully, revealing no bugs in our \specification.

\begin{figure*}[t]
	\subfloat[FSCQ suite]{
		\includegraphics[width=.33\textwidth]{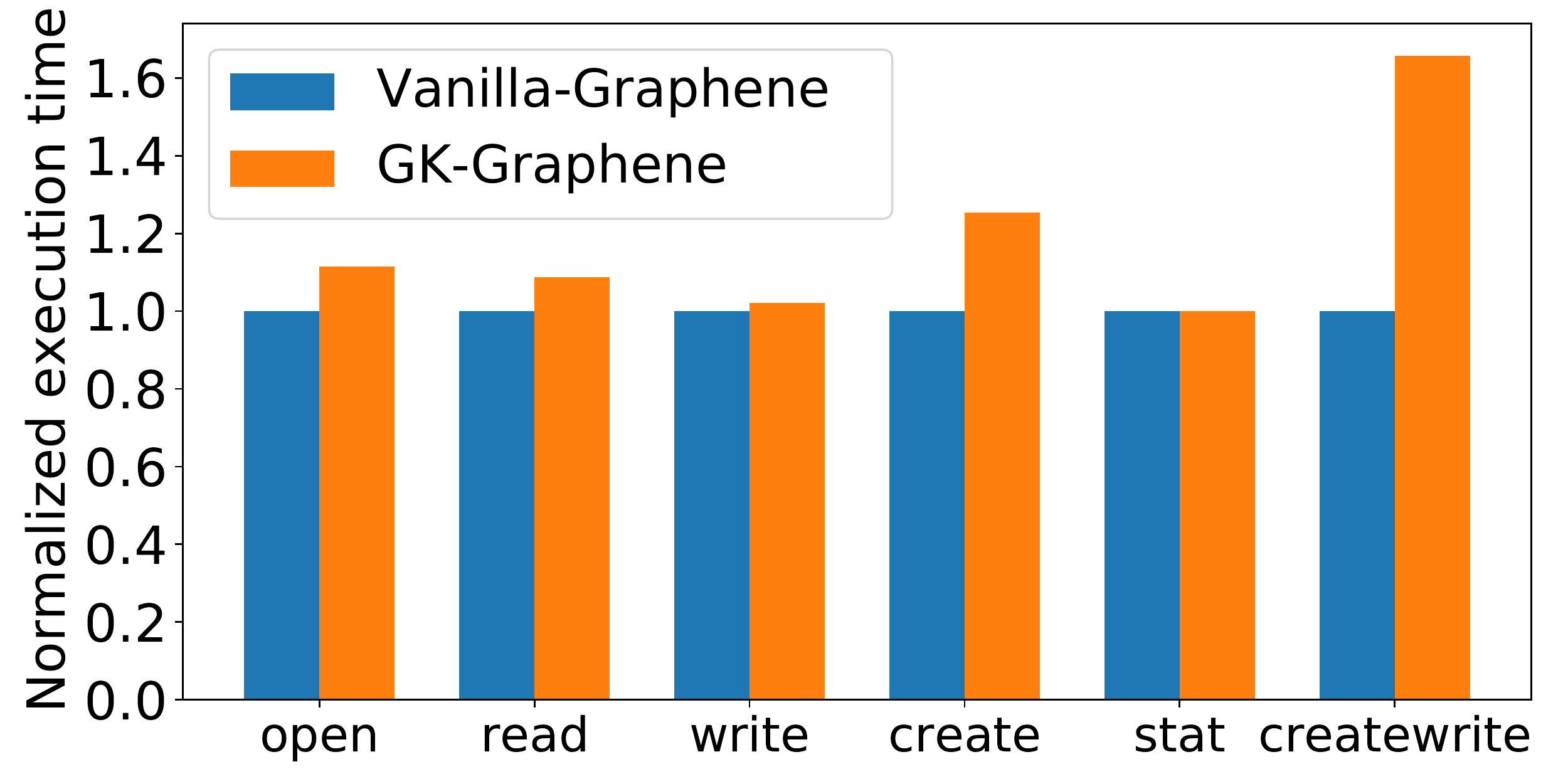}
		\label{fig:fscq}
	}
	\subfloat[FS overheads with SQLite]{
		\includegraphics[width=0.33\textwidth]{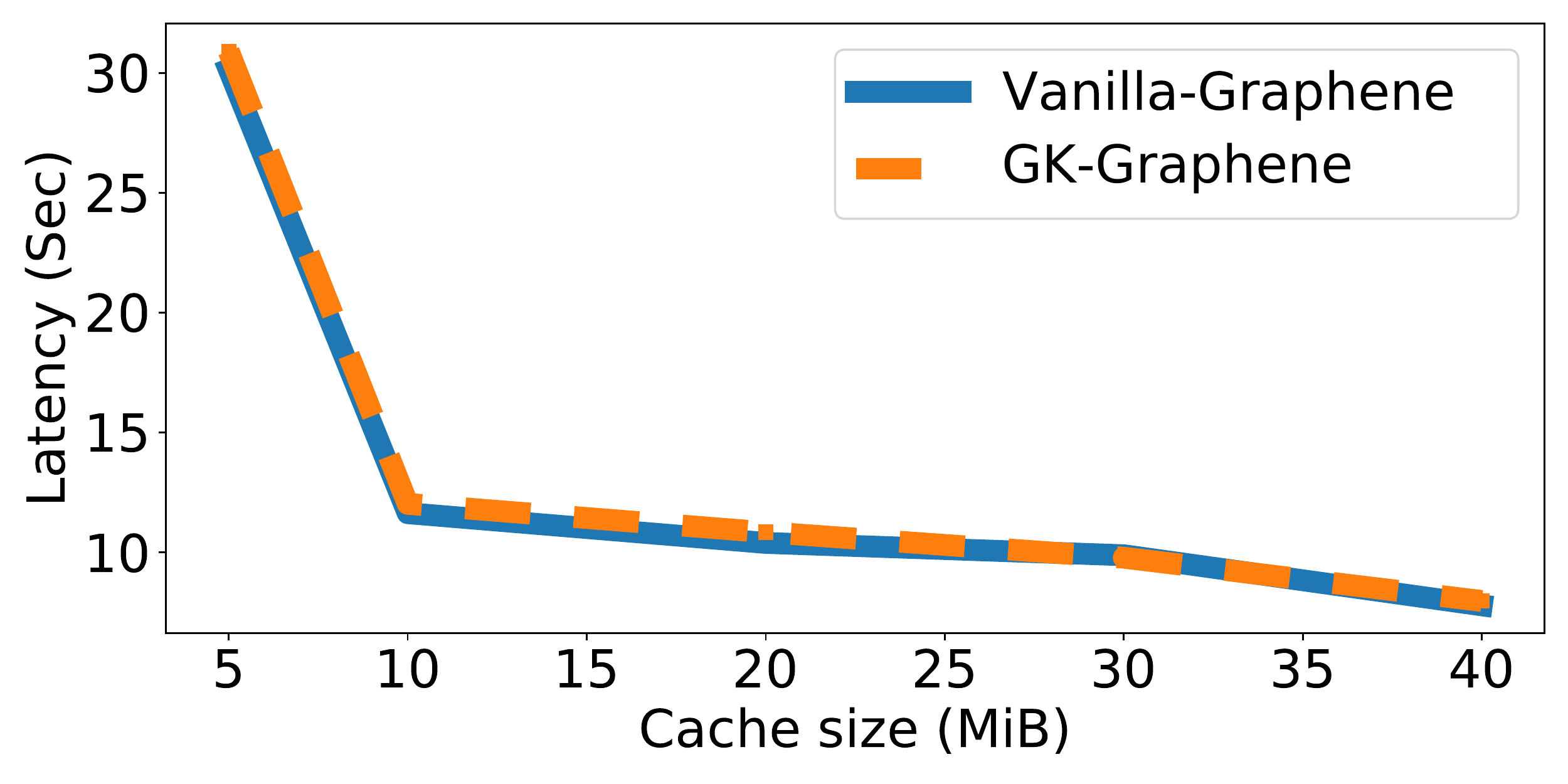}
		\label{fig:sqlite}
	}
	\subfloat[Futex overheads with Memcached]{
		\includegraphics[width=.33\textwidth]{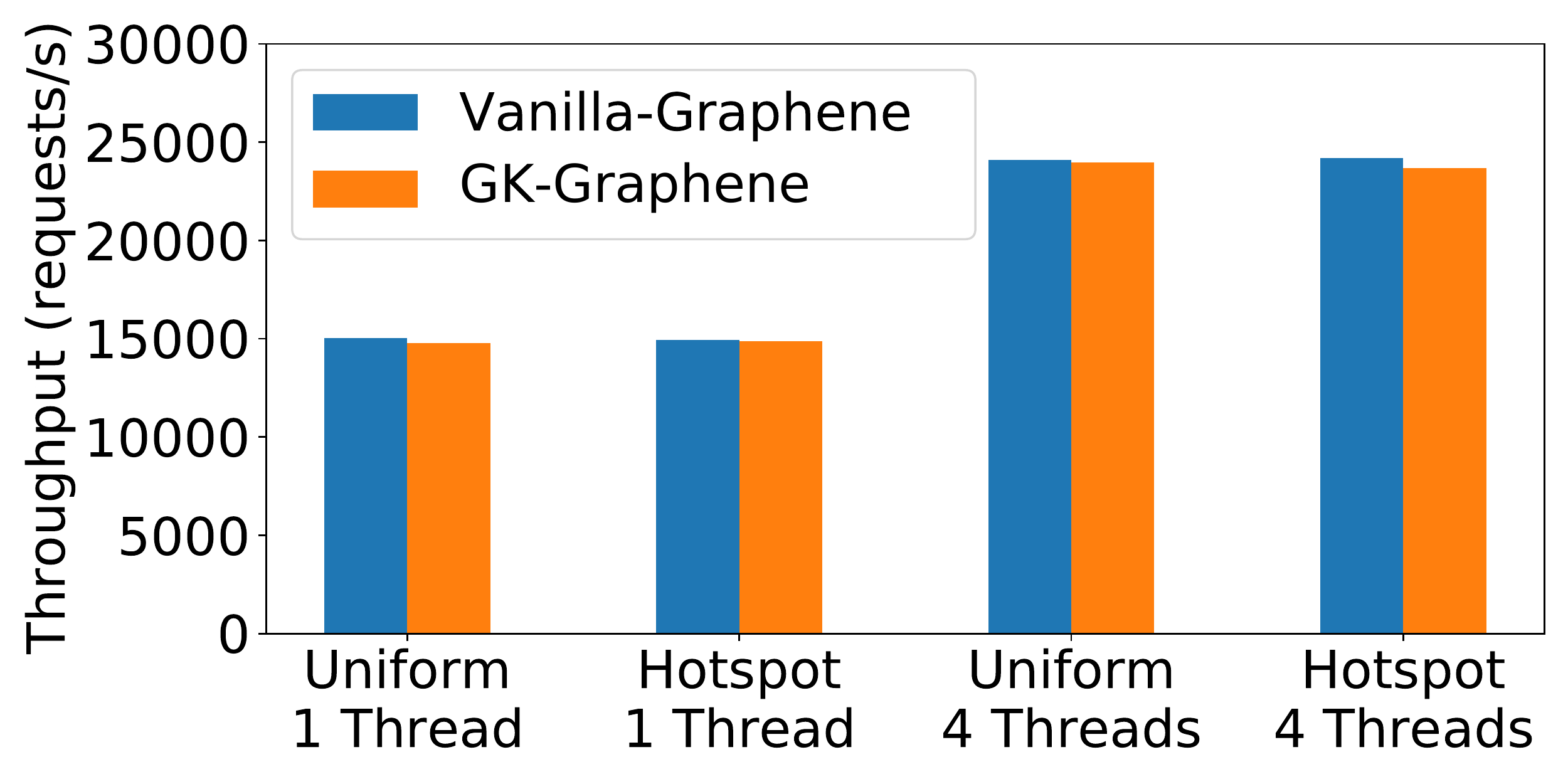}
		\label{fig:memcached}
	}

	\caption{FS and futex overheads: \GrapheneUs vs \GrapheneVanilla.}
\end{figure*}

\ctable[
caption=Fuzzing workloads and vulnerabilities found by \projname
(by Emilia~\cite{emilia:ndss21}),
label=tbl:fuzz-workloads
]{p{.06\columnwidth}p{.15\columnwidth}p{.15\columnwidth}p{.15\columnwidth}p{.20\columnwidth}}{}{ \FL
	& curl & memcached & Redis & zlib 
	\ML
	ver. & 7.58 & 1.5.20 & 5.0.5 & 1.2.11
	\NN
	LOC & 130k & 18k & 115k & 18k
	\NN
	Desc. & Web client & MemKVS & DB KVS & Library
	\NN
	Vuln. & 5 (3) & 4 (4) & 1 (1) & 2 (1)
	\LL
}

\subsection{Vulnerability checking}
\label{sec:eval:fuzzing}
We demonstrate the use of \BackendFuzz to find vulnerabilities in 
applications and libOSs.

\paragraph{Library OS vulnerabilities}
We run our fuzzer against Graphene-SGX using a simple test
we implement, which performs all file operations supported in Graphene-SGX 
and asserts that the results are as expected.
This test is sufficiently simple to allow us to find vulnerabilities in the libOS itself. 

Our fuzzer invokes the test, each time replacing a system call return value or 
output parameter with our generated adversarial values. Overall, we run our fuzzer with
4,000 different malicious values. 
Fortunately, Graphene-SGX successfully protects against most of these attacks. 
However, we find one vulnerability in their protection logic
that results in illegal memory access (responsibly disclosed to developers). 
Such vulnerabilities are known to compromise entire enclaves~\cite{rop_sgx:usenixsec17}.
Finally, we validate that this vulnerability is \emph{correctly detected by our validator} when  
integrated into Graphene-SGX. 

\paragraph{Comparison with Emilia}
We use the dedicated fuzzer Emilia~\cite{emilia:ndss21}
to detect invalid memory accesses in applications due to Iago attacks.
We restrict the fuzzing only to FS-related system calls 
and use applications previously fuzzed by Emilia (the malicious values used by Emilia
are publicly available). We exclude the git application because its run was too long.
Before fuzzing we use \BackendFuzz to generate malicious values tailored to these applications.
Last, we compare fuzzing with Emilia's values with those generated by \BackendFuzz.

We manually validate each crash and filter duplicates.
We report the results in Table~\ref{tbl:fuzz-workloads}. 
\BackendFuzz identifies \emph{a superset} of vulnerabilities as compared 
to Emilia. Specifically, it finds identical vulnerabilities for all the tests, plus 
new ones in zlib and curl that Emilia did not find.

\subsection{Microbenchmarks and end-to-end performance}
\label{sec:eval:validator_overhreads}
Finally, we study the performance overhead of our in-enclave runtime validation.
We integrate both FS and futex validators into Graphene-SGX. 
We disable the validator's internal file-content integrity validation and use Graphene's 
file protection via authenticated encryption in both configurations.
We call  Graphene-SGX with the validator as
\GrapheneUs and to the original version as \GrapheneVanilla.
In the following benchmarks, we execute the 
same tests with the same setup both in \GrapheneVanilla
and \GrapheneUs. 

\paragraph{FS Stress Test}
We use FSCQ to measure performance overheads of common FS operations,
as done in previous work~\cite{besfs:usenixsec20, fscq:atc16}. 
We modify each test in the suite to measure end-to-end execution time.
This puts the focus on the FS operations and excludes initialization time. 
Before each test, we add a setup step to validate that the FS operations succeed. 
Finally, we modify the write microbenchmark to perform sequential writes 
to avoid internal caches in Graphene-SGX, 
which forces all write operations to be forwarded to the OS.

The results are shown in Figure~\ref{fig:fscq}.  
We observe that the overheads of most operations are relatively small.
Specifically, note that with \GrapheneUs
\id{read} and \id{write} incur less than 10\%
overhead compared to \GrapheneVanilla. 
The file creation overhead is higher, about 60\%, as expected because of the complexity of the file path validations. 

\paragraph{Mutex Stress Test}
We use the mutex stress test from LTP to measure the performance
of our \emph{futex} \specification. We configure the test to run for 20
seconds. 
We observe negligible overheads: 4,238,450 and 4,227,715
lock/unlock operations performed in \GrapheneVanilla and \GrapheneUs, respectively.
This is expected: futex validations are simple comparisons and small compared to the
sleep times under high lock contention. 

\paragraph{SQLite: FS overheads}
SQLite is a popular database used for libOS evaluation in prior work~\cite{cosmix:atc19,sgxbounds:eurosys17}. 
We measure end-to-end latency with \id{speedtest1} shipped with SQLite v3.34.1
while varying the cache size. 
We use a single thread and disable mmap-based access to
the DB files.
The results are shown in Figure~\ref{fig:sqlite}.
We observe that \GrapheneUs is on a par with \GrapheneVanilla,
regardless of the cache size. This demonstrates that 
the FS validator overheads are negligible.

\paragraph{Memcached: futex overheads}
Memcached~\cite{memcached} is a popular
key-value store that was used in prior work on enclaves~\cite{eleos:eurosys17,cosmix:atc19,scone:osdi16,glamdring:usenixatc17}. 
We evaluate Memcached
using the YCSB workload generator~\cite{ycsb:socc10} with
the predefined \id{workload C} as in previous work~\cite{glamdring:usenixatc17,autarky:eurosys20}. 
It performs 100\% random GET operations for 1\,KB records. 
To  focus specifically on futex overheads, we evaluate 
Memcached with one or four serving threads, with uniform access and hotspot. In hotspot, we define 5\%
of the entries as a hot set with an access probability of 95\%.
This creates high and low contention on internal mutexes.

We co-locate YCSB and Memcached on the same machine and 
pin each one to a different core to avoid network overheads.
We preload Memcached with 1 MB of data before each test 
to avoid SGX paging overheads~\cite{eleos:eurosys17}. 
We report the maximum throughput achieved.
The results are shown in Figure~\ref{fig:memcached}. Exactly as in the mutex stress tests,
we observe negligible overheads in both high- and low-contention setups. 

\subsection{In-enclave software size comparison}
We measure the code size of in-enclave FS implementation
in SGX-LKL and compare it to the size of our FS \specification,
compiler and trusted library (6,173 lines in total).
Note that these components constitute the FS validator's TCB 
as it is generated by them.
Both protect against Iago attacks. 
As SGX-LKL includes a full Linux kernel, it includes many FS 
implementations. For completeness, we measure
both the FS directory in LKL (887k LOC), and
only for ext4 that SGX-LKL mounts by default (36k LOC).
\projname minimizes the TCB, 
while still protecting against Iago attacks. Yet,
admittedly, our FS model does not support all the ext4 features.

%% file: related.tex
\section{Related Work}
\label{sec:related}

\paragraph{TEEs}
TEEs~\cite{sgx:hasp13:mckeen,sgx:hasp13:anati, 
amd_sev:white_paper, amd_sev_snp:white_paper,
keystone:eurosys20,tdx:white_paper,cca:white_paper} protects against a
privileged adversary. 
TEEs do not protect against Iago attacks on the
untrusted software interface, which is the focus of 
\projname.

\paragraph{Enclave interface attacks}
\citeauthor{tale_of_two_worlds:ccs19}~\cite{tale_of_two_worlds:ccs19}
provided a comprehensive analysis of enclave interface attacks,
generalizing Iago attacks from system calls to ocalls 
and manually inspecting libOSs code
to find Iago vulnerabilities.
\projname, finds such vulnerabilities 
automatically via \BackendFuzz.

\paragraph{Interface defenses}
TeeRex~\cite{teerex:usenixsec20} detects memory-related vulnerabilities
via symbolic execution. The Intel \id{edger8r} tool generates
ocalls that enforce type safety, and ~\citeauthor{sep_logic}~\cite{sep_logic}
used separation logic to validate that pointers are entirely in/out
of the enclave address space. 
COIN attacks~\cite{coin_attacks:asplos20} further generalize
Iago attacks to the unexpected invocation of enclave functions 
and use symbolic execution to detect such vulnerabilities.
SGXPecial~\cite{sgxpecial:eurosec21} restricts valid control flows across 
the interface to mitigate code reuse exploits.
Unlike these defenses, \projname validates the \emph{complete 
semantics of cross-interface calls} and detects any deviation
from them.

Many runtime systems attempt to mitigate Iago attacks~\cite{
inktag:asplos13, sego:asplos16, graphenesgx:usenix17, civet:usenixsec20,
haven:osdi14, sgxlkl, scone:osdi16, ghostrider:asplos15}.
We focus on Inktag and Sego, as TEE runtimes were already 
discussed (\S\ref{sec:background}). Both 
use the trust-but-verify approach, as in \projname. 
However, they rely on a trusted hypervisor that has
access to devices and can therefore verify the OS behavior 
when accessing a certain file, which facilitates the validation of FS operations.
\projname cannot rely on a trusted hypervisor but 
still shows it is possible to verify OS services from within
the enclave.

\paragraph{Iago fuzzers}
Emilia~\cite{emilia:ndss21} provides a system call fuzzer to
detect Iago vulnerabilities. Unlike Emilia, which 
performs static analysis on programs' source code, \projname 
generates malicious values based on the \specification.

\paragraph{Services \specifications}
Modeling of services for verification
was extensively studied~\cite{besfs:usenixsec20, yggdrasil:osdi16, fscq:atc16, hyperkernel:sosp17, lightweight_formal:sosp21, sel4:sosp09, certikos:osdi16}.
Unlike previous \specifications that were used to verify
service implementations, \projname generates 
validation, mock, and vulnerability checking tools.

%% file: conclusion.tex
\section{Conclusion}
\label{sec:conclusion}

\projname is a framework that 
facilitates model development without expertise 
in formal methods. The model is used to systematically protect 
trusted code from untrusted services. 
We develop and refine \specifications for FS and synchronization primitives 
and use them to protect Memcached and SQLite executing in SGX enclaves.
We also find vulnerabilities in Graphene-SGX and production applications 
with the FS vulnerability checker. 
We believe that \projname is not limited to TEEs, 
and forms a foundation for protecting general untrusted services,
e.g., for microkernel
and remote untrusted services.

\projname's source code will be publicly available.